\documentclass[submission,copyright,creativecommons]{eptcs}
\providecommand{\event}{SNR 2017} 
\usepackage{breakurl}             
\usepackage{underscore}           

\usepackage[english]{babel}
\usepackage[T1]{fontenc}
\usepackage{amsmath, amsfonts, amsthm, epsfig, xspace, bm, mathrsfs}
\usepackage{algorithm,algorithmic}
\usepackage{multimedia}
\usepackage{xcolor}
\usepackage{array} 
\usepackage{calc}
\usepackage{enumitem}
\usepackage{graphicx}
\usepackage{hyperref}

\newcommand{\bx}{x}
\newcommand{\by}{y}
\newcommand{\bz}{z}
\newtheorem{remark}{Remark}
\newtheorem{theorem}{Theorem}
\newtheorem{corollary}{Corollary}

\newtheorem{proposition}{Proposition}

\newtheorem{definition}{Definition}
\newcommand{\bff}{f}

\title{Control Synthesis of Nonlinear Sampled Switched Systems using Euler's Method}
\author{A. Le Co\"ent\footnote{corresponding author} \quad F. De Vuyst
\institute{CMLA, CNRS \& ENS Paris-Saclay}
\email{lecoent@cmla.ens-cachan.fr}
\and
L. Chamoin
\institute{LMT, CNRS \& ENS Paris-Saclay}
\and
L. Fribourg
\institute{LSV, CNRS, INRIA \& ENS Paris-Saclay}
}
\def\titlerunning{Control synthesis of switched systems using Euler method}
\def\authorrunning{A. Le Co\"ent, F. De Vuyst, L. Chamoin \& L. Fribourg}

\graphicspath{{figures/}}

\begin{document}
\maketitle

\begin{abstract}
In this paper, we propose a symbolic control synthesis method for nonlinear sampled switched systems whose vector fields are {\em one-sided Lipschitz}. The main idea is to use an approximate model obtained from the {\em forward Euler method} to build a guaranteed control. The benefit of this method is that the error introduced by symbolic modeling is bounded by choosing suitable time and space discretizations. The method is implemented in the interpreted language Octave. Several examples of the literature are performed and 
the results are compared with results obtained with a previous method based on the Runge-Kutta integration method.
\end{abstract}

\section{Introduction}
As said in \cite{NL_minimator},
in the methods of symbolic analysis and control of
hybrid systems, the way of representing sets of state values
and computing reachable sets for systems defined by
ordinary differential equations (ODEs) is fundamental
(see, e.g., \cite{Althoff2013a,girard2005reachability}).
An interesting approach appeared recently, based on the
propagation of reachable sets using guaranteed Runge-Kutta
methods with adaptive step size control (see \cite{BMC12,immler2015verified}). In
\cite{NL_minimator} such guaranteed
integration methods are used in the framework of {\em sampled switched systems}.

Given an ODE of the form $\dot{x}(t)=f(t,x(t))$, and a {\em set} of initial values $X_0$,
a symbolic (or ``set-valued'') integration method consists in computing a sequence of
approximations $(t_n, \tilde{x}_n)$ of the solution $x(t; x_0)$ of the ODE
with $x_0\in X_0$ such that $\tilde{x}_n \approx x(t_n; x_{n-1})$.
Symbolic integration methods extend classical {\em numerical} integration methods which correspond to the case where $X_0$ is just a singleton $\{x_0\}$.
The simplest numerical method is Euler's method in which $t_{n+1} = t_n + h$
for some step-size $h$ and $\tilde{x}_{n+1} = \tilde{x}_n + h f(t_n,\tilde{x}_n)$; so
the derivative of $x$ at time $t_n$, $f(t_n, x_n)$, is used as an
approximation of the derivative on the whole time interval. This method is very simple
and fast, but requires small step-sizes $h$.
More advanced
methods coming from the Runge-Kutta family use a few
intermediate computations to improve the approximation
of the derivative. The general form of an explicit $s$-stage
Runge-Kutta formula
of the form $\tilde{x}_{n+1}=\tilde{x}_n+h\Sigma_{i=1}^sb_ik_i$
where $k_i=f(t_n+c_ih, \tilde{x}_n+h\Sigma_{j=1}^{i-1}a_{ij}k_j)$
for $i=2,3,...,s$.
A challenging question
is then to compute a bound on the distance between
the true solution and the numerical solution, i.e.:
$\|x(t_n; x_{n-1}) - x_n\|$. This distance is associated to the {\em local
truncation error} of the numerical method.
In~\cite{NL_minimator}, 
such a bound is computed using the {\em Lagrange
remainders} of Taylor expansions.
This is achieved using {\em affine arithmetic} \cite{AffineA97}
(by application of the Banach's fixpoint theorem and 
Picard-Lindel\"of operator, see \cite{Nedialkov}). 
In the end, the Runge-Kutta based method of \cite{NL_minimator} is an elaborated method
that requires the use of affine arithmetic, 
Picard iteration and 
computation of Lagrange remainder.

In contrast, in this paper, we use ordinary arithmetic
(instead of affine arithmetic)
and a basic Euler scheme (instead of Runge-Kutta schemes). We
need neither estimate Lagrange remainders nor perform Picard iteration
in combination with Taylor series.
Our simple Euler-based approach is made possible  by having 
recourse to the notion of {\em one-sided Lipschitz} (OSL) function
\cite{Donchev98}.
This allows us to bound directly the {\em global error},
i.e. the distance between the approximate point~$\tilde{x}(t)$
computed by the Euler scheme
and the exact solution $x(t)$ for all $t\geq 0$
(see Theorem~\ref{th:1}).

{\bf Plan.}
In Section \ref{sec:rw}, we give  details on related work.
In Section \ref{sec:OSL}, we state our main result that bounds
the global error introduced by the Euler scheme in the context of systems
with OSL flows.
In Section~\ref{sec:appl}, we explain how to apply this result
to the synthesis of symbolic control of sampled switched systems.
We give numerical experiments and results
in Section \ref{sec:experiment} for five exampes of the literature,
and compare them with results obtained with the method of \cite{NL_minimator}.
We give final remarks in Section \ref{sec:fr}.
\section{Related work}\label{sec:rw}

Most of the recent work on the symbolic (or set-valued) integration of nonlinear ODEs is based
on the upper bounding of the Lagrange remainders either in the framework of 
Taylor series or Runge-Kutta schemes 
\cite{Althoff2013a,BCD13,BMC12,CAS12,chen2013flow,NL_minimator,Makino2009,report,dit2016validated}.
Sets of states are generally represented as vectors of intervals 
(or ``rectangles'')
and are manipulated  through interval arithmetic \cite{Moore66}
or affine arithmetic \cite{AffineA97}.
Taylor expansions with Lagrange remainders are also used in the work
of \cite{Althoff2013a}, which uses ``polynomial zonotopes'' 
for representing sets of states in addition to interval vectors.
None of these works uses
the Euler scheme nor the notion of one-sided Lipschitz constant.

In the literature on symbolic integration, the Euler scheme with OSL conditions
is explored in \cite{Donchev98,Lempio95}.
Our approach is similar but establishes an {\em analytical} result for the global error of Euler's estimate
(see Theorem \ref{th:1}) rather than analyzing, in terms of complexity,
the speed of convergence to zero, the accuracy and the stability of Euler's method.

In the control literature, OSL conditions
have been recently applied to control and stabilization 
\cite{Abbaszadeh2010,Cai2015},
but do not make use of Euler's method.
To our knowledge, our work applies for the first time Euler's scheme
with OSL conditions to the symbolic
control of hybrid systems.

\section{Sampled switched systems with OSL conditions}\label{sec:OSL}
\subsection{Control of switched systems}\label{ss:OSC}
Let us consider the nonlinear switched system
\begin{equation}
 \dot \bx(t) = f_{\sigma (t)}(\bx(t))
 \label{eq:sys}
\end{equation}
defined for all $t \geq 0$, where $\bx(t) \in \mathbb{R}^n$ is the state
of the system, $\sigma(\cdot) : \mathbb{R}^+ \longrightarrow U$ is the
switching rule. The finite set $U = \{ 1, \dots , N \}$ is the set of
switching {\em modes} of the system.  We focus on {\em sampled switched systems}:
given a sampling period $\tau >0$, switchings will occur at times
$\tau$, $2\tau$, \dots{} The switching rule~$\sigma(\cdot)$ is thus constant
on the time interval $\lbrack (k-1) \tau , k \tau )$ for $k \geq 1$.
For all $j\in U$, $f_j$ is a function from $\mathbb{R}^n$ to~$\mathbb{R}^n$.
We make the following hypothesis:
$$(H0)\quad  \mbox{ For all  $j\in U$, $f_j$ is a locally Lipschitz continuous map}.$$
As in \cite{girard2010approximately}, we make the assumption that the vector field $f_j$ is such that the solutions of the differential equation  (\ref{eq:sys}) are defined, e.g. by assuming that the support of the vector field $f_j$ is compact.
We will denote by $\phi_\sigma(t;\bx^0)$ the 
solution at time~$t$ of the system:

\begin{equation}
\begin{aligned}
  \dot \bx(t) & = \bff_{\sigma (t)}(\bx(t)), \\
  \bx(0) & =  \bx^0. \\
\end{aligned}
 \label{eq:sampled-sys}
\end{equation}

Often, we will consider $\phi_\sigma(t;x^0)$ on the interval $0\leq t<\tau$ for which $\sigma(t)$ is equal to a constant, say $j\in U$. In this case, we will abbreviate $\phi_\sigma(t;x^0)$ as $\phi_j(t;x^0)$. We will also consider $\phi_\sigma(t;x^0)$ on the interval $0\leq t<k\tau$
where $k$ is a positive integer, and 
$\sigma(t)$ is equal to a constant, say $j_{k'}$,
on each interval
$[(k'-1)\tau,k'\tau)$ with $1\leq k'\leq k$; in this case,
we will abbreviate $\phi_\sigma(t;x^0)$ as $\phi_\pi(t;x^0)$,
where $\pi$ is a sequence of $k$ modes (or ``pattern'') of the form
$\pi=j_1\cdot j_2\cdot\dots\cdot j_k$. 

We will assume that $\phi_\sigma$ is {\em continuous} at time $k\tau$ for all positive integer $k$.
This means that there is no ``reset'' at time $k'\tau$ ($1\leq k'\leq k$);
the value of 
$\phi_\sigma(t,x^0)$ for $t\in[(k'-1)\tau,k\tau]$
corresponds to the solution of $\dot{x}(u)=f_{j_{k'}}(x(u))$  for $u\in [0,\tau]$
with initial value $\phi_\sigma((k'-1)\tau;x^0)$.

Given a ``recurrence set'' $R\subset\mathbb{R}^n$ and a ``safety set'' $S\subset\mathbb{R}^n$ which contains $R$ ($R\subseteq S$), we are interested in the synthesis of a control such that:
starting from any initial point $x\in R$, the controlled trajectory always returns to $R$ within a bounded time while never leaving $S$. We suppose that sets $R$ and $S$ are compact. Furthermore, we suppose that $S$ is convex.



We denote by $T$ a compact overapproximation of the image by $\phi_j$ of $S$ for $0\leq t\leq \tau$ and $j\in U$, i.e.~$T$ is such that
$$ T\supseteq \{\phi_j(t;x^0) \ |\ j\in U, 0\leq t\leq\tau, x^0\in S\}.$$
The existence of $T$ is guaranteed by assumption $(H0)$. We know furthermore 
by $(H0)$ that, for all $j\in U$, there exists a constant $L_j>0$ such that:
\begin{equation}
\| \bff_j(\by)-\bff_j(\bx) \| \leq L_j \, \|\by-\bx\|\quad \forall \bx,\by\in S.
\label{eq:lipschitz}
\end{equation}
Let us define $C_j$ for all $j\in U$:
\begin{equation}
C_j = \sup_{x\in S}\  L_j\|f_j(x)\|
\quad
\text{for all} \quad j\in U.
\label{eq:L}
\end{equation}
We make the additional hypothesis 
that the mappings $f_j$ are {\em one-sided Lipschitz} (OSL)
\cite{Donchev98}.
Formally:
$$(H1) 
\quad \mbox{ For all $j \in U$, there exists a constant $\lambda_j\in \mathbb{R}$ such that}$$ 
\[
\langle \bff_j(\by)-\bff_j(\bx), \by-\bx \rangle \leq \lambda_j\, \|\by-\bx\|^2\quad
\forall \bx,\by\in T,
\]
where $\langle \cdot, \cdot\rangle$ denotes the scalar product of two vectors of $\mathbb{R}^n$.
%

\vspace{1em}

\begin{remark}
Constants $\lambda_j$, $L_j$ and $C_j$  ($j\in U$) can 
be computed using (constrained) optimization algorithms. See
Section \ref{sec:experiment} for details.
\end{remark}


\subsection{Euler approximate solutions}

Given an  initial point $\tilde{x}^0\in S$ and a mode $j\in U$, we define the following ``linear approximate solution''~$\tilde{\phi}_j(t;\tilde{x}^0)$ for $t$ on $[0,\tau]$ by:
\begin{equation}
\tilde{\phi}_j(t;\tilde{x}^0) = \tilde \bx^0 + t f_j(\tilde \bx^0).
\label{eq:grossier}
\end{equation}
\newline
\vspace{1em}
Note that formula~\eqref{eq:grossier} is nothing else but the explicit forward Euler scheme with
``time step'' $t$. It is thus a {consistent} approximation of order $1$ in $t$
of the exact solution of~\eqref{eq:sys}
under the hypothesis $\tilde \bx^0=\bx^0$.

More generally, given an initial point $\tilde{x}^0\in S$ and  pattern $\pi$ of $U^k$,
we can define a ``(piecewise linear) approximate solution'' 
$\tilde{\phi}_\pi(t;\tilde{x}^0)$ of $\phi_\pi$ at time $t\in[0,k\tau]$ as follows:
\begin{itemize}
\item $\tilde{\phi}_\pi(t;\tilde{x}^0) = t f_j(\tilde{x}^0) + \tilde{x}^0$ if $\pi=j\in U$, $k=1$ and $t\in[0,\tau]$, and

\item $\tilde{\phi}_{\pi}(k\tau+t;\tilde{x}^0) = t f_j(\tilde{z}) + \tilde{z}$
with $\tilde{z}=\tilde{\phi}_{\pi'}((k-1)\tau;\tilde{x}^0)$, if $k\geq 2$, 
$t\in[0,\tau]$,
$\pi=j\cdot \pi'$ for some $j\in U$ and $\pi'\in U^{k-1}$.
\end{itemize}


We wish to synthesize a guaranteed control $\sigma$ for $\phi_{\sigma}$
using the approximate functions $\tilde{\phi}_\pi$.
We define the closed ball of center $x\in\mathbb{R}^n$ and radius $r>0$, denoted $B(x,r)$, as the set $\{x'\in\mathbb{R}^n \ |\ \|x'-x\| \leq r\}$.

Given a positive real $\delta$, we now define the expression $\delta_j(t)$
which, as we will see in Theorem \ref{th:1}, represents (an upper bound on)
the error associated to $\tilde{\phi}_j(t; \tilde{x}^0)$
(i.e. $\|\tilde{\phi}_j(t; \tilde{x}^0)-\phi_j(t; x^0)\|$).

\begin{definition}\label{def:4}
Let $\delta$ be a positive constant. Let us define, for all $0\leq t\leq \tau$,
$\delta_j(t)$ as follows:
\begin{itemize}
\item  if $\lambda_j <0$:
$$\delta_j(t)=\left(\delta^2 e^{\lambda_j t}+
 \frac{C_j^2}{\lambda_j^2}\left(t^2+\frac{2 t}{\lambda_j}+\frac{2}{\lambda_j^2}\left(1- e^{\lambda_j t} \right)\right)\right)^{\frac{1}{2}}$$

\item if $\lambda_j = 0:$
$$\delta_j(t)= \left( \delta^2 e^{t} + C_j^2 (- t^2 - 2t + 2 (e^t - 1)) \right)^\frac{1}{2}$$


\item if $\lambda_j > 0:$
$$\delta_j(t)=\left(\delta^2 e^{3\lambda_j t}+
\frac{C_j^2}{3\lambda_j^2}\left(-t^2-\frac{2t}{3\lambda_j}+\frac{2}{9\lambda_j^2}
\left(e^{3\lambda_j t}-1\right)\right)\right)^{\frac{1}{2}}$$
\end{itemize}
\end{definition}

Note that $\delta_j(t)=\delta$ for $t=0$. 
The function $\delta_j(\cdot)$ depends implicitly on two parameters: $\delta\in\mathbb{R}$ and $j\in U$. In Section \ref{sec:appl}, we will use the notation $\delta'_j(\cdot)$
where the parameters are denoted by $\delta'$ and $j$.

\begin{theorem}\label{th:1}

Given
a sampled switched system satisfying (H0-H1), consider
a point $\tilde{x}^0$
and a positive real $\delta$. 
We have,
for all $x^0\in B(\tilde{x}^0,\delta)$, $t\in [0,\tau]$ and  $j\in U$:

$\phi_j(t;x^0)\in B(\tilde{\phi}_j(t;\tilde{x}^0),\delta_j(t))$.
\end{theorem}

\begin{proof}

Consider on $t\in [0,\tau]$ the differential equations 
\[
\frac{d  \bx(t)}{dt} = \bff_j(x(t))
\]
and
\[
\frac{d \tilde \bx(t)}{dt} = \bff_j(\tilde \bx^0).
\]
with initial points $x^0\in S,\tilde{x}^0\in S$ respectively.
We will abbreviate $\phi_j(t;x^0)$ 
(resp. $\tilde{\phi}_j(t;\tilde{x}^0)$) as $x(t)$ (resp.~$\tilde{x}(t)$).
We have
\[
\frac{d}{dt}(\bx(t)-\tilde\bx(t)) =  \left( \bff_j(\bx(t))-\bff_j(\tilde\bx^0)\right),
\]
then
\begin{eqnarray*}
\frac{1}{2}\, \frac{d}{dt}(\|\bx(t)-\tilde \bx(t)\|^2) &=& 
\left\langle \bff_j(\bx(t))-\bff_j(\tilde\bx^0),  \bx(t)-\tilde \bx(t) \right\rangle \\ 
&=& \left\langle \bff_j(\bx(t))-\bff_j(\tilde \bx(t))+\bff_j(\tilde \bx(t))
                      -\bff_j(\tilde\bx^0),  \bx(t)-\tilde \bx(t) \right\rangle\\
&=& \left\langle \bff_j(\bx(t))-\bff_j(\tilde \bx(t)), \bx(t)-\tilde \bx(t) \right\rangle
+\left\langle \bff_j(\tilde \bx(t))                      -\bff_j(\tilde\bx^0),  \bx(t)-\tilde \bx(t) \right\rangle\\
&\leq& \left\langle \bff_j(\bx(t))-\bff_j(\tilde \bx(t)), \bx(t)-\tilde \bx(t) \right\rangle
+ \| \bff_j(\tilde \bx(t))                      -\bff_j(\tilde\bx^0)\| \| \bx(t)-\tilde \bx(t) \|.
\end{eqnarray*}
The last expression has been obtained using the Cauchy-Schwarz inequality. Using $(H1)$ and (\ref{eq:lipschitz}), we have 
%
\begin{eqnarray*}
\frac{1}{2}\, \frac{d}{dt}(\|\bx(t)-\tilde \bx(t)\|^2)
&\leq& \lambda_j \|\bx(t)-\tilde\bx(t)\|^2
+\, \| \bff_j( \tilde \bx(t))- \bff_j(\tilde\bx^0)\|\,  \|\bx(t)-\tilde \bx(t)\|
 \\
&\leq& \lambda_j \|\bx(t)-\tilde\bx(t)\|^2
+L_j\, \|\tilde \bx(t)-\tilde\bx^0\|\,  \|\bx(t)-\tilde \bx(t)\| \\
&\leq& \lambda_j \|\bx(t)-\tilde\bx(t)\|^2
+L_j t\, \|\bff_j(\tilde\bx^0)\|\,  \|\bx(t)-\tilde \bx(t)\|. \\
\end{eqnarray*}
Using \eqref{eq:L} and a Young inequality, we then have
\begin{eqnarray*}
\frac{1}{2}\, \frac{d}{dt}(\|\bx(t)-\tilde \bx(t)\|^2)
&\leq& \lambda_j \|\bx(t)-\tilde\bx(t)\|^2
+C_j\,t\, \|\bx(t)-\tilde \bx(t)\| \\ 
&\leq& \lambda_j \|\bx(t)-\tilde\bx(t)\|^2
+C_j\,t\, \frac{1}{2}
\left( \alpha \|\bx(t)-\tilde \bx(t)\|^2 + \frac{1}{ \alpha } 
\right)
\end{eqnarray*}
for all $ \alpha  >0$.

\begin{itemize}
 \item In the case $\lambda_j <0$:

For $t>0$, we choose $ \alpha >0$ such that
$C_j t \alpha  = -\lambda_j$, 
i.e.
$ \alpha  = -\frac{\lambda_j}{C_j\, t}$.
It follows, for all $t\in [0,\tau]$:
\[
\frac{1}{2}\, \frac{d}{dt}(\|\bx(t)-\tilde \bx(t)\|^2) \leq
\frac{\lambda_j}{2} \|\bx(t)-\tilde\bx(t)\|^2 - \frac{C_j t}{2 \alpha }
=\frac{\lambda_j}{2} \|\bx(t)-\tilde\bx(t)\|^2 - \frac{(C_j t)^2}{2\lambda_j}.
\]
We thus get:
%
\[
\|\bx(t)-\tilde \bx(t)\|^2 \leq \|\bx^0-\tilde \bx^0\|^2\, e^{\lambda_j t}
+ \frac{C_j^2}{\lambda_j^2}\left(t^2+\frac{2 t}{\lambda_j}+\frac{2}{\lambda_j^2}\left(1- e^{\lambda_j t} \right)\right).
\]

 
%

 \item In the case $\lambda_j >0$:
 
For $t>0$, we choose $ \alpha >0$ such that
$C_j t \alpha  = \lambda_j$, 
i.e.
$ \alpha  = \frac{\lambda_j}{C_j\, t}$.
It follows, for all $t\in [0,\tau]$:
\[
\frac{1}{2}\, \frac{d}{dt}(\|\bx(t)-\tilde \bx(t)\|^2) \leq
\frac{3\lambda_j}{2} \|\bx(t)-\tilde\bx(t)\|^2 + \frac{C_j t}{2 \alpha }
=\frac{3\lambda_j}{2} \|\bx(t)-\tilde\bx(t)\|^2 + \frac{(C_j t)^2}{2\lambda_j}.
\]
We thus get:
%
\[
\|\bx(t)-\tilde \bx(t)\|^2 \leq \|\bx^0-\tilde \bx^0\|^2\, e^{3\lambda_j t}+
\frac{C_j^2}{3\lambda_j^2}\left(-t^2-\frac{2t}{3\lambda_j}+\frac{2}{9\lambda_j^2}
\left(e^{3\lambda_j t}-1\right)\right)\]


 
%

\item In the case $\lambda_j =0$:

For $t>0$, we choose $ \alpha = \frac{1}{C_j t}$. It follows:
$$\frac{d}{dt}(\|\bx(t)-\tilde \bx(t)\|^2)
 \leq   \|\bx(t)-\tilde \bx(t)\|^2 + C_jt^2
 $$
 
We thus get:
$$ \|x(t)-\tilde{x}(t)\|^2 \leq \|x^0-\tilde{x}^0\|^2 e^{t} + C_j^2 (- t^2 - 2t + 2 (e^t - 1)) $$

In every case, since by hypothesis $x^0\in B(\tilde{x}^0,\delta)$ (i.e. $\| x^0 - \tilde{x}^0\|^2 \leq \delta^2$),
we have, for all $t\in [0,\tau]$:
\[
\|x(t)-\tilde{x}(t)\| \leq 
\delta_j(t).
 \]

 \end{itemize}

 It follows: $\phi_j(t;x^0)\in B(\tilde{\phi}_j(t;\tilde{x}^0), \delta)$ for $t\in [0,\tau]$.

\end{proof}

\vspace{1em}

\begin{remark}
In Theorem \ref{th:1}, we have supposed that the step size $h$ used in Euler's method was equal to the sampling period $\tau$ of the switching system.
Actually, in order to have better approximations, it is often convenient
to take a {\em fraction} of $\tau$ as for $h$ (e.g., $h=\frac{\tau}{10}$).
Such a splitting is called ``sub-sampling'' in numerical methods.
See
Section \ref{sec:experiment} for details.
\end{remark}

\vspace{1em}

\begin{corollary}\label{cor:1}

Given a sampled switched system 
satisfying (H0-H1), consider
a point $\tilde{x}^0\in S$, a real $\delta>0$ and a mode $j\in U$ such that:
\begin{enumerate}
\item $B(\tilde{x}^0,\delta)\subseteq S$,
\item $B(\tilde{\phi}_j(\tau;\tilde{x}^0),\delta_j(\tau))\subseteq S$, and
\item $\frac{d^2(\delta_j(t))}{dt^2}>0$ for all $t\in [0,\tau]$.
\end{enumerate}
Then we have, for all $x^0\in B(\tilde{x}^0,\delta)$ and $t\in[0,\tau]$:\ 
%
$\phi_j(t;x^0)\in S$.

\end{corollary}

\begin{proof}
By items 1 and 2, 
$B(\tilde{\phi}_j(t;\tilde{x}^0),\delta_j(t))$ for $t=0$ and $t=\tau$.
Since $\delta_j(\cdot)$ is convex on~$[0,\tau]$ by item~3, and $S$ is convex,
we have $B(\tilde{\phi}_j(t;\tilde{x}^0),\delta_j(t))\subseteq S$
for all $t\in[0,\tau]$. It follows from Theorem \ref{th:1}
that $\phi_j(t;x^0)\in B(\tilde{\phi}_j(t;\tilde{x}^0),\delta_j(t))\subseteq S$
for all $1\leq t\leq\tau$.
\end{proof}

\vspace{1em}

\begin{remark}
Condition 3 of Corollary \ref{cor:1} on the convexity of $\delta_j(\cdot)$ 
on $[0,\tau]$ can be established again using an optimization function
(see Section \ref{sec:experiment}).
\end{remark}


\section{Application to control synthesis}\label{sec:appl}

Consider a point~$\tilde{x}^0\in S$, a positive real  $\delta$ 
and a pattern $\pi$ of length $k$.  
Let $\pi(k')$ denote the $k'$-th element (mode) of~$\pi$ for $1\leq k'\leq k$.
Let us abbreviate
the $k'$-th approximate point
$\tilde{\phi}_{\pi}(k'\tau;\tilde{x}^{0})$ 
as~$\tilde{x}_\pi^{k'}$ for $k'=1,...,k$, 
and let $\tilde{x}_\pi^{k'}=\tilde{x}^0$ for $k'=0$. It is easy to show that
$\tilde{x}_\pi^{k'}$ can be defined recursively for $k'=1,...,k$, by:
$\tilde{x}_\pi^{k'}=\tilde{x}_\pi^{k'-1}+\tau f_{j}(\tilde{x}_\pi^{k'-1})$
with $j=\pi(k')$.

Let us now
denote by  $\delta_\pi^{k'}$
(an upper bound on) the error associated to $\tilde{x}_\pi^{k'}$,
i.e. $\|\tilde{x}_\pi^{k'}- \phi_\pi(k'\tau;x^0)\|$.
Using repeatedly Theorem \ref{th:1},
$\delta_{\pi}^{k'}$ can be defined recursively
as follows:

For $k'=0$: $\delta_{\pi}^{k'}=\delta$,
and for $1\leq k'\leq k$: $\delta^{k'}_{\pi}=\delta'_j(\tau)$
where $\delta'$ denotes $\delta^{k'-1}_{\pi}$, and $j$ denotes
$\pi(k')$.\\
Likewise, for $0\leq t\leq k\tau$, let us
denote by $\delta_{\pi}(t)$
(an upper bound on) the 
global error associated to $\tilde{\phi}_\pi(t;\tilde{x}^0)$
(i.e.~$\|\tilde{\phi}_\pi(t;\tilde{x}^0)- \phi_\pi(t;x^0)\|$).
Using Theorem \ref{th:1},
$\delta_{\pi}(t)$ can be defined itself as follows:
\begin{itemize}
\item for $t=0$:\ $\delta_\pi(t)=\delta$,
\item for $0<t\leq k\tau$:\ 
$\delta_{\pi}(t)=\delta'_{j}(t')$ with 
$\delta'=\delta_\pi^{\ell-1}$, $j=\pi(\ell)$,
$t'=t-(\ell-1)\tau$ and
$\ell=\lceil \frac{t}{\tau}\rceil$.
%
\end{itemize}
Note that, for $0\leq k'\leq k$, we have: 
$\delta_{\pi}(k'\tau)=\delta_\pi^{k'}$. We have:

\begin{theorem}
Given a sampled switched system satisfying (H0-H1),
%
consider a point~$\tilde{x}^0\in S$, a positive real  $\delta$ 
and a pattern $\pi$ of length $k$ such that, for all $1\leq k'\leq k$:
\begin{enumerate}
\item $B(\tilde{x}_\pi^{k'}, \delta_{\pi}^{k'}) \subseteq S$ and
\item $\frac{d^2(\delta'_j(t))}{dt^2}>0$ for all $t\in [0,\tau]$, with $j=\pi(k')$ and $\delta'=\delta_\pi^{k'-1}$.
\end{enumerate}
Then we have, for all $x^0\in B(\tilde{x}^0,\delta)$ and $t\in [0,k\tau]$:\ \ 
$\phi_{\pi}(t;x^0)\in S$.
\label{propbis:1}
\end{theorem}
\begin{proof}
By induction on $k$ using Corollary~\ref{cor:1}.
%
\end{proof}

The statement of Theorem \ref{propbis:1} is illustrated in Figure \ref{fig:tube} for $k=2$. From Theorem \ref{propbis:1}, it easily follows:
\begin{figure}[h]
\centering
 \includegraphics[width=0.4\linewidth,clip,trim=1cm 0cm 1cm 0cm]{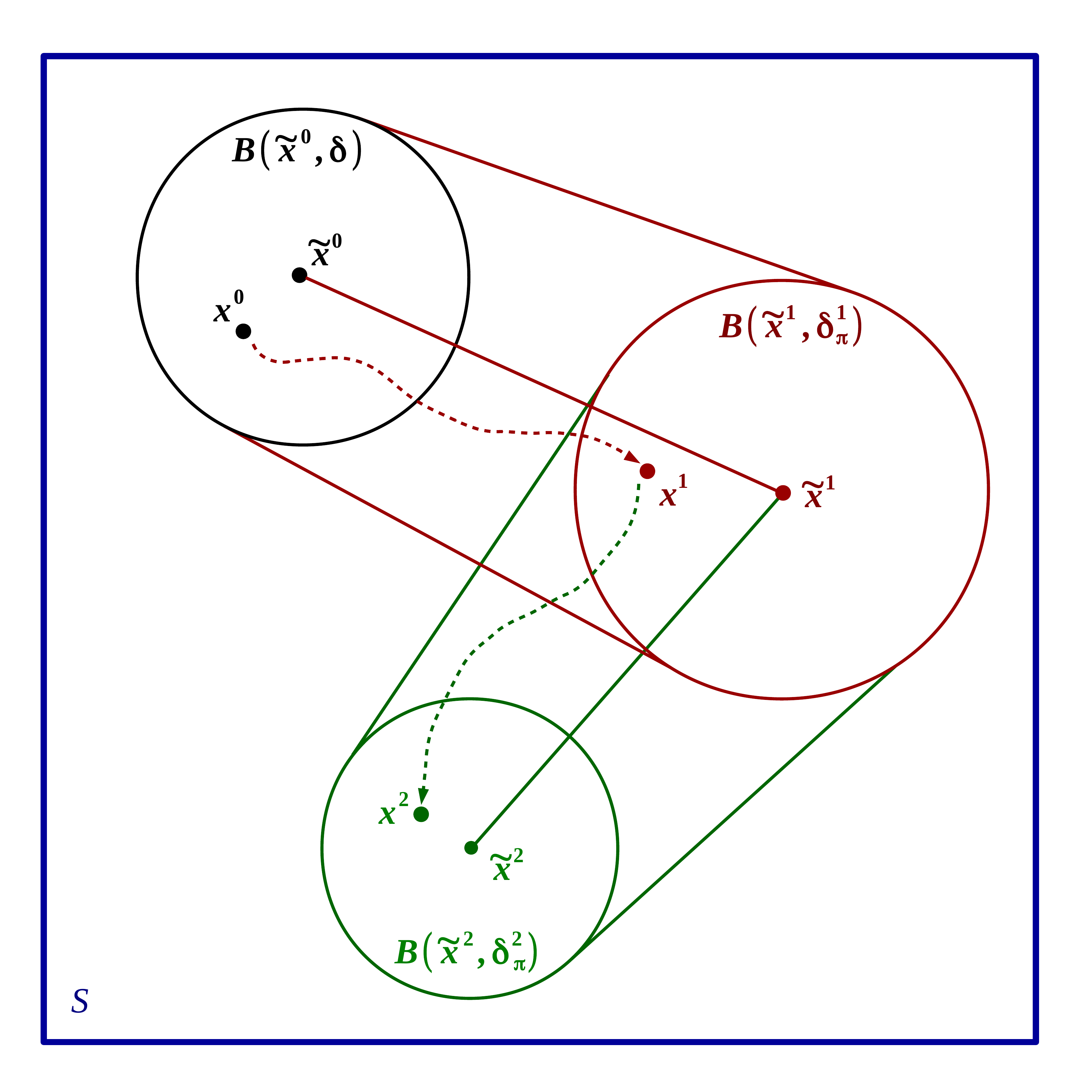}
\caption{Illustration of Theorem \ref{propbis:1}.}
 \label{fig:tube}
\end{figure}

\begin{corollary}
Given a switched system satisfying (H0-H1), consider
a positive real $\delta$ 
and a finite set of points
$\tilde{x}_1,\dots\tilde{x}_m$ of $S$ such that all the balls $B(\tilde{x}_i,\delta)$ 
cover $R$ and are included into~$S$ (i.e. $R\subseteq \bigcup_{i=1}^mB(\tilde{x}_i,\delta)\subseteq S$).
Suppose furthermore that, for all $1\leq i\leq m$, there exists a pattern $\pi_i$ of length $k_i$ such that:
\begin{enumerate}
\item $B((\tilde{x}_i)_{\pi_i}^{k'},\delta_{\pi_i}^{k'}) \subseteq S$,
for all $k'=1,\dots,k_i-1$
\item 
$B((\tilde{x}_i)_{\pi_i}^{k_i}, \delta_{\pi_i}^{k_i}) \subseteq R.$
\item $\frac{d^2(\delta'_j(t))}{dt^2}>0$ 
with $j=\pi_i(k')$ and $\delta'=\delta_{\pi_i}^{k'-1}$, for all 
$k'\in\{1,...,k_i\}$ and $t\in [0,\tau]$.
\end{enumerate}
These properties induce a control
$\sigma$\footnote{Given an initial point $x\in R$, the induced control $\sigma$ corresponds to a sequence
of patterns $\pi_{i_1},\pi_{i_2},\dots$ defined as follows:
Since  $x\in R$, there exists a 
a point $\tilde{x}_{i_1}$  with $1\leq i_1\leq m$ such that $x\in B(\tilde{x}_{i_1},\delta)$; then using pattern $\pi_{i_1}$, one has: $\phi_{\pi_{i_1}}(k_{i_1}\tau;x)\in R$. Let $x'=\phi_{\pi_{i_1}}(k_{i_1}\tau;x)$; there exists a point $\tilde{x}_{i_2}$ with $1\leq i_2\leq m$ such that $x'\in B(\tilde{x}_{i_2},\delta)$, etc.}
which guarantees
\begin{itemize}
\item (safety): if $x\in R$, then $\phi_{\sigma}(t;x) \in S$ for all $t\geq 0$,
and
\item (recurrence):
if $x\in R$ then $\phi_{\sigma}(k\tau;x)\in R$ for some $k\in\{k_1,\dots,k_m\}$.
\end{itemize}
\label{propter:1}
\end{corollary}

Corollary \ref{propter:1} gives the theoretical foundations of the following method for synthesizing $\sigma$ ensuring recurrence in $R$ and safety in $S$:
\begin{itemize}
\item we (pre-)compute $\lambda_j, L_j, C_j$ for all $j\in U$;
\item we find $m$ points $\tilde{x}_1,\dots\tilde{x}_m$ of $S$
and $\delta>0$ such that $R\subseteq \bigcup_{i=1}^m B(\tilde{x}_i,\delta)\subseteq S$;
\item we find $m$ patterns $\pi_i$  ($i=1,...,m$)
such that conditions 1-2-3 of Corollary \ref{propter:1} are satisfied.
\end{itemize} 
A covering of $R$ with balls as stated in Corollary \ref{propter:1} is illustrated in Figure \ref{fig:tiling}.
The control synthesis method based on~Corollary \ref{propter:1}
is illustrated in Figure \ref{fig:post} (left)
together with an illustration of method of \cite{NL_minimator} (right).

\begin{figure}[h]
\centering
 \includegraphics[width=0.4\linewidth,clip,trim=1cm 0cm 1cm 0cm]{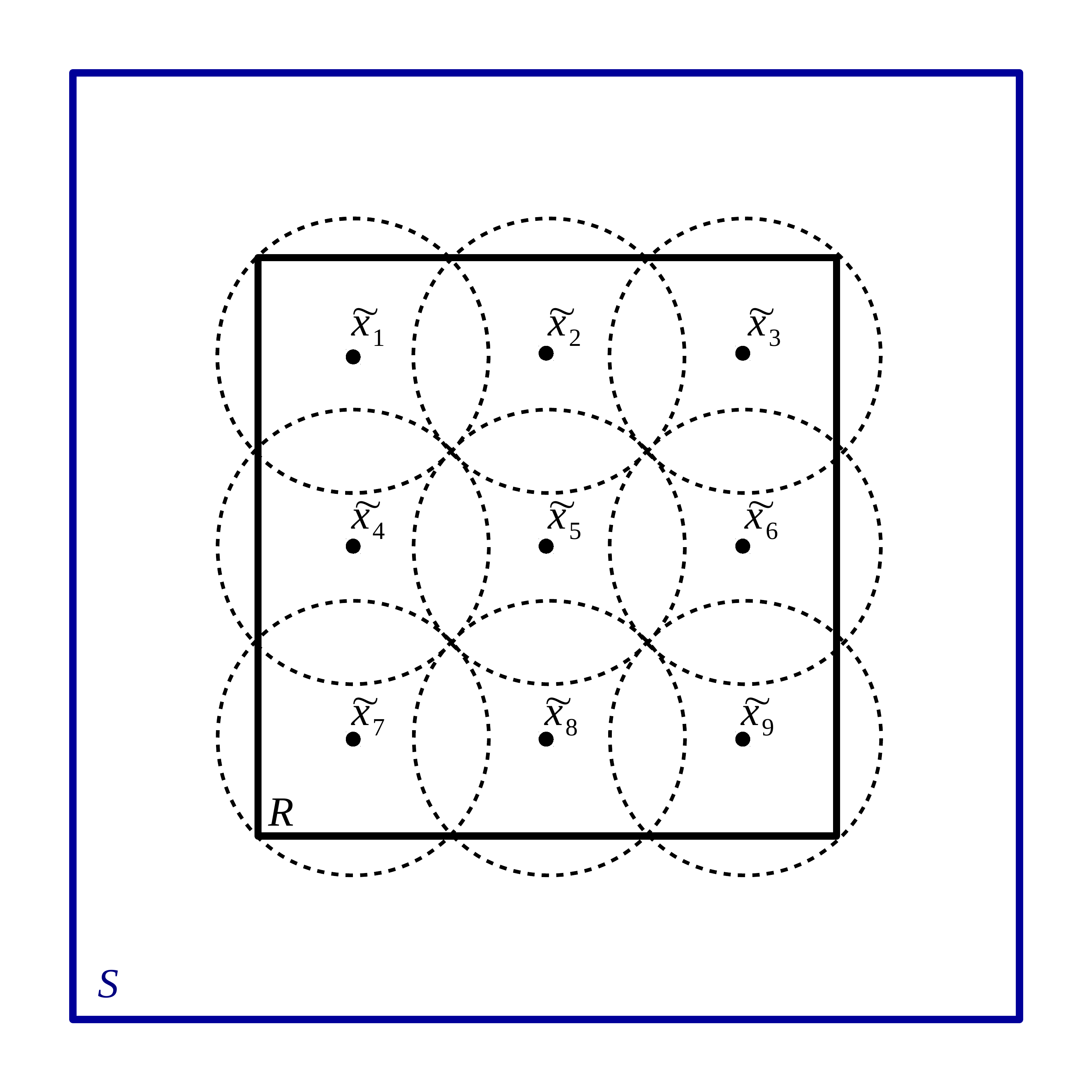}
\caption{A set of balls covering $R$ and contained in $S$.}
\label{fig:tiling}
\end{figure}


\begin{figure}[h]
\centering
\begin{tabular}{cc}
 \includegraphics[width=0.4\linewidth,clip,trim=1cm 0cm 1cm 0cm]{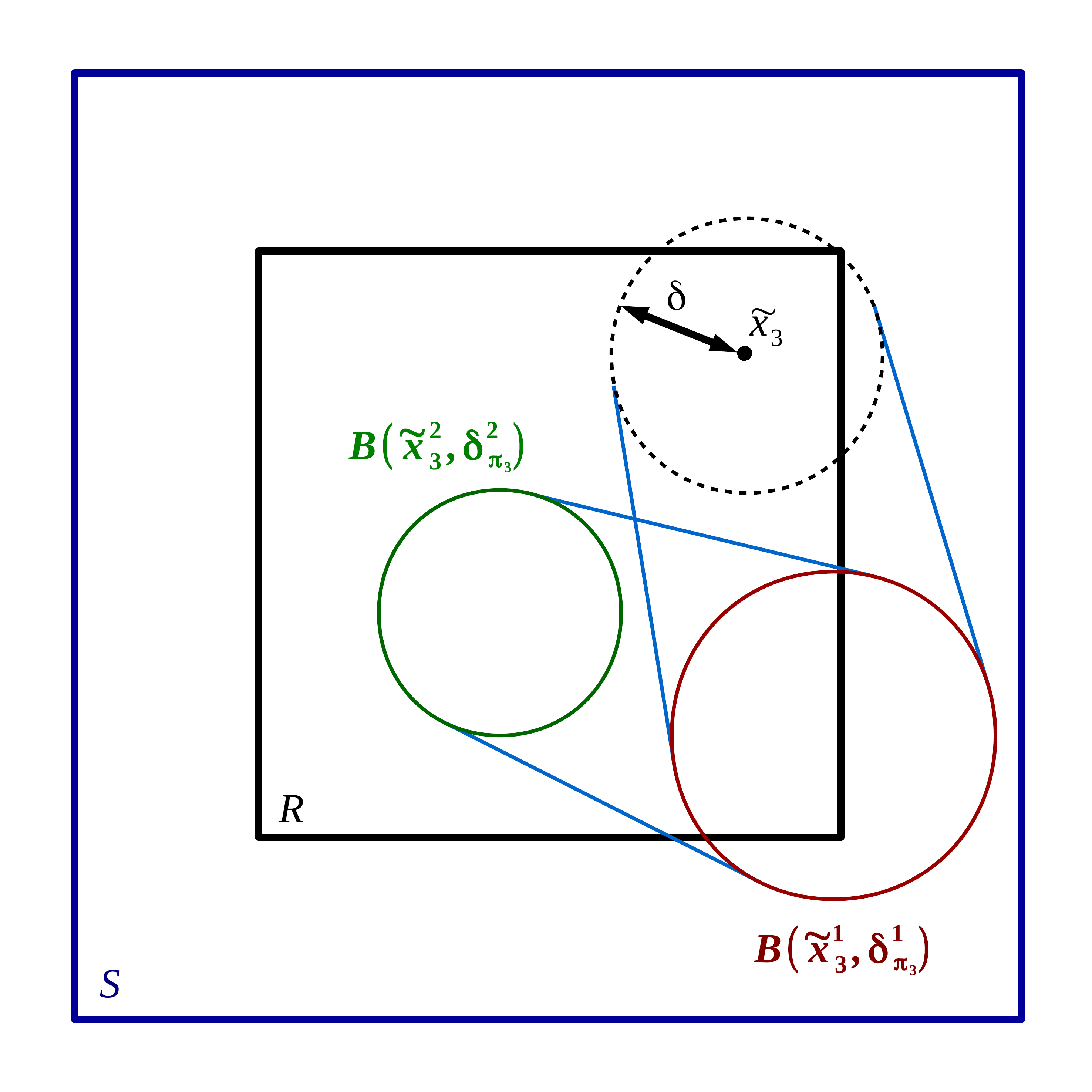}
&
 \includegraphics[width=0.4\linewidth,clip,trim=1cm 0cm 1cm 0cm]{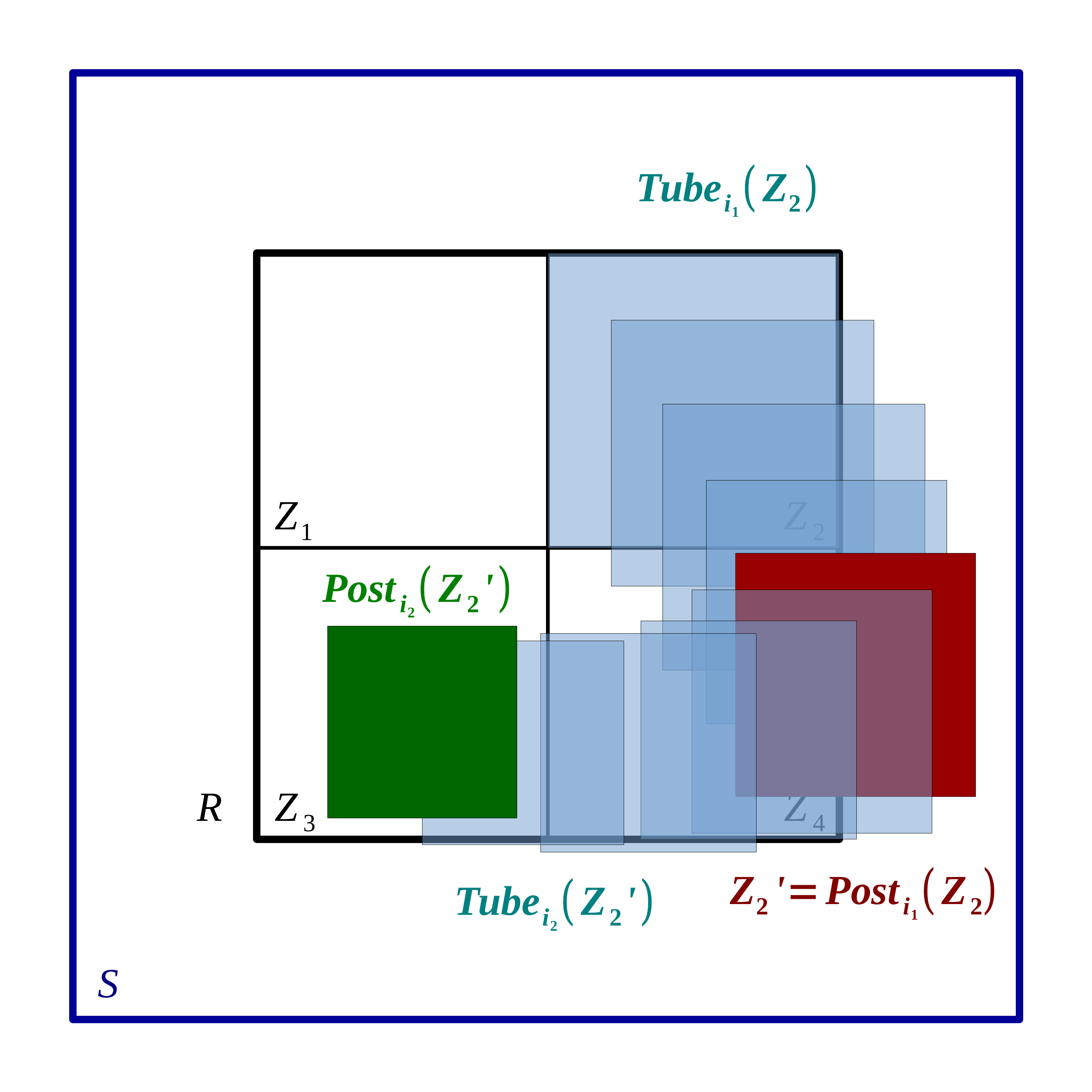}
\end{tabular}
\caption{Control of ball $B(\tilde x_3,\delta)$ with our method (left);  
control of tile $Z_2$ with the method of \cite{NL_minimator}~(right).}
\label{fig:post}
\end{figure}

\section{Numerical experiments and results}
\label{sec:experiment}
This method has been implemented in the interpreted language Octave, and the experiments performed on a 2.80 GHz Intel Core i7-4810MQ CPU with 8 GB
of memory.

The computation of constants $L_j$, $C_j$, $\lambda_j$ 
($j\in U$) are realized with
a constrained optimization algorithm.
They are performed using the ``sqp'' function of Octave, applied on the following 
optimization problems:
\begin{itemize}
 \item Constant $L_j$:
 $$ L_j = \max_{{x,y}\in S,\  x\neq y} \frac{\| f_j(y) - f_j(x) \|}{\| y - x \|}  $$
 \item Constant $C_j$:
 $$ C_j = \max_{{x}\in S} L_j \| f_j(x) \|$$
 \item Constant $\lambda_j$:
 $$ \lambda_j = \max_{{x,y}\in T,\  x\neq y} \frac{\langle f_j(y) - f_j(x), y - x \rangle}{\|y - x \|^2 }$$
 \end{itemize}

Likewise, the convexity test
$\frac{d^2(\delta'_j(t))}{dt^2}>0$ can be performed similarly.

Note that in some cases, it is advantageous to use a time sub-sampling to compute the image of a ball.
Indeed, because of the exponential growth of the radius $\delta_j(t)$ within time,
computing a sequence of balls can lead to smaller ball images. 
It is particularly advantageous when a constant $\lambda_j$ is negative.
We illustrate this with the example of the DC-DC converter. 
It has two switched modes, for which we have $\lambda_1 = -0.014215$
and $\lambda_2 = 0.142474$. 
In the case $\lambda_j < 0$, the associated formula $\delta_j(t)$ has the behavior 
of Figure \ref{fig:delta_t_pos} (a). 
In the case $\lambda_j > 0$, the associated formula $\delta_j(t)$ has the behavior 
of Figure \ref{fig:delta_t_pos} (b). In the case $\lambda_j < 0$, 
if the time sub-sampling is small enough, 
one can compute a sequence of balls with reducing radius, which makes the synthesis easier. 

\begin{figure}[h]
\centering
\begin{tabular}{cc}
\includegraphics[width=0.45\linewidth,clip,trim=0.7cm 0cm 1.0cm 0cm]{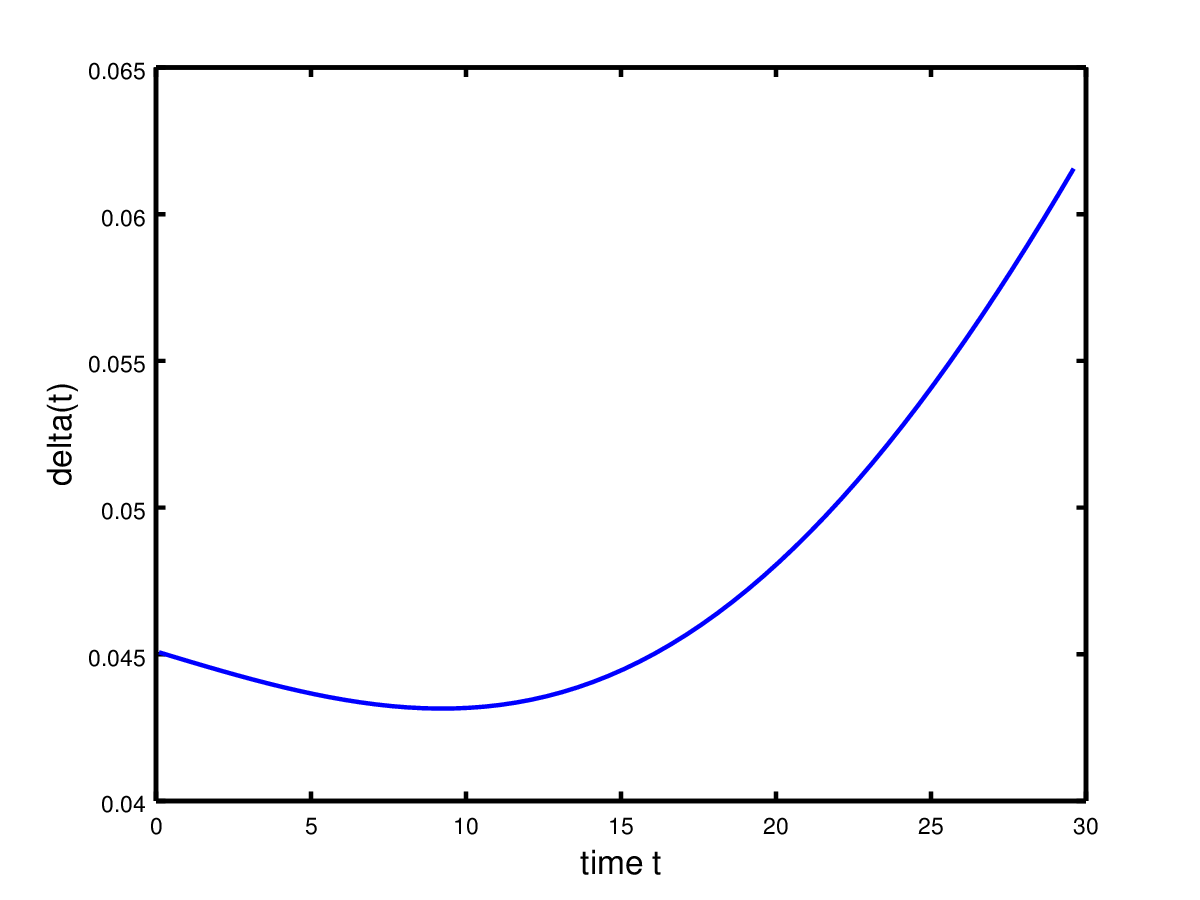} & \includegraphics[width=0.45\linewidth,clip,trim=0.7cm 0cm 1.0cm 0cm]{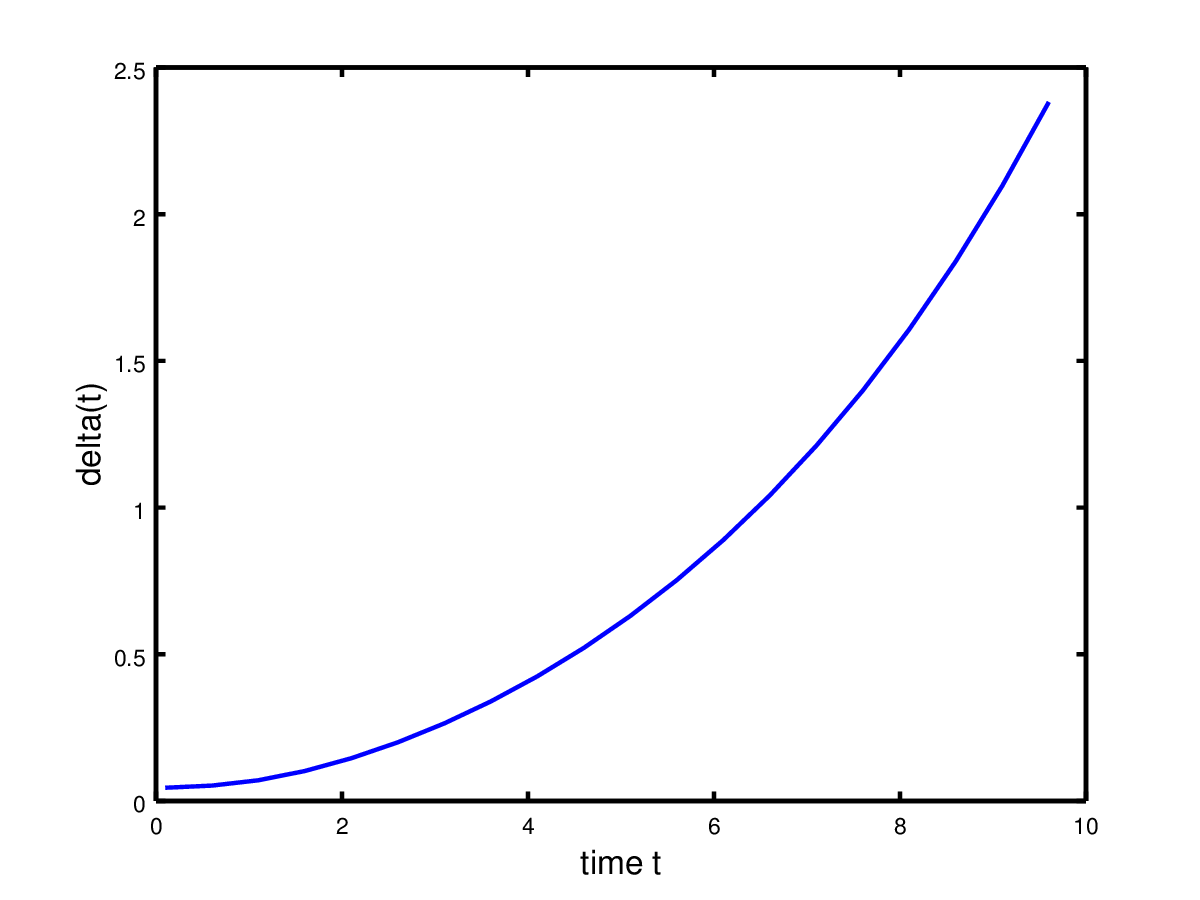} \\
(a) & (b)
\end{tabular}
 \caption{Behavior of $\delta_j(t)$ for the DC-DC converter with $\delta_j(0) = 0.045$. 
 (a) Evolution of $\delta_1(t)$ (with $\lambda_1<0$); 
 (b) Evolution of $\delta_2(t)$ (with $\lambda_2>0$). }
  \label{fig:delta_t_pos}
 \end{figure}

In the following, we give the results obtained with our Octave implementation 
of this Euler-based method on 5 examples, and compare them with those 
given by the C++ implementation {\em DynIBEX} \cite{dynibex} of the Runge-Kutta based 
method used in \cite{NL_minimator}.

 \subsection{Four-room apartment}
 We describe a first application on a 4-room 16-switch building ventilation case study adapted from
\cite{meyer:tel-01232640}. The model has been simplified in order to get constant
parameters.
The system is a four room apartment
subject to heat transfer between the rooms, with the external
environment, the underfloor, and human beings.  The dynamics
of the system is given by the following equation:
\begin{equation*}
 \frac{d T_i}{dt} = \sum_{j \in \mathcal{N}^\text{*} \setminus \{i\}} a_{ij} (T_j -
 T_i) + \delta_{s_i} b_i (T_{s_i}^4 - T_i ^4 )  + c_i
 \max\left(0,\frac{V_i - V_i^\text{*}}{\bar{ V_i} -
   V_i^{\text{*}}}\right)(T_u - T_i), \quad \mbox{for } i=1,...,4.
\end{equation*}

The state of the system is given by the temperatures in the rooms
$T_i$, for $i \in \mathcal{N} = \{ 1 , \dots , 4 \}$.  Room~$i$ is
subject to heat exchange with different entities stated by the indices
$\mathcal{N}^\text{*} = \{1,2,3,4,u,o,c \}$.
We have $T_0=30, T_c=30, T_u=17$, $\delta_{s_i}=1$ for $i\in\mathcal{N}$.
The (constant) parameters $T_{s_i}$, $V_i^\text{*}$, $\bar V_i$, $a_{ij}$, $b_i$,
$c_i$ are given in~\cite{meyer:tel-01232640}. 
%
The control input is $V_i$ ($i \in \mathcal{N}$).
In the experiment, $V_1$ and $V_4$ can take the values $0$V
or $3.5$V, and $V_2$ and~$V_3$ can take the values $0$V or $3$V. This
leads to a system of the form~\eqref{eq:sys} with $\sigma(t) \in U =\{
1, \dots, 16 \}$, the $16$ switching modes corresponding to the
different possible combinations of voltages $V_i$.  
The sampling period is $\tau = 30$s.
Compared simulations are given in Figure \ref{fig:simu}.
On this example, the Euler-based method works better than {\em DynIBEX}
in terms of CPU time.

 \begin{table}[h]
 \centering
\begin{tabular}{|c|c|c|}
   \hline 
   &\multicolumn{1}{c|}{Euler} & \multicolumn{1}{c|}{DynIBEX} \\
   \hline
   $R$ & \multicolumn{2}{c|}{$[20,22]^2\times[22,24]^2$} \\
   $S$ & \multicolumn{2}{c|}{$[19,23]^2\times[21,25]^2$} \\   
\hline
$\tau$ & \multicolumn{2}{c|}{30} \\
\hline
Time subsampling & No & \\   
   \hline
 Complete control & Yes  & Yes \\
\hline
$\max_{j= 1, \dots,16} \lambda_j$  &  $-6.30\times 10^{-3}$   &        \\
$\max_{j= 1, \dots,16} C_j$  &  $4.18\times 10^{-6}$ &                     \\
\hline
Number of balls/tiles & 4096 & 252 \\
Pattern length & 1 & 1 \\
\hline
CPU time &  63 seconds & 249 seconds\\ \hline
  \end{tabular}
\label{table:4M}
\caption{Numerical results for the four-room example.}
 \end{table}

\begin{figure}[h]
\centering
\begin{tabular}{cc}
 \includegraphics[width=0.4\linewidth,clip,trim=1cm 0cm 1cm 0cm]{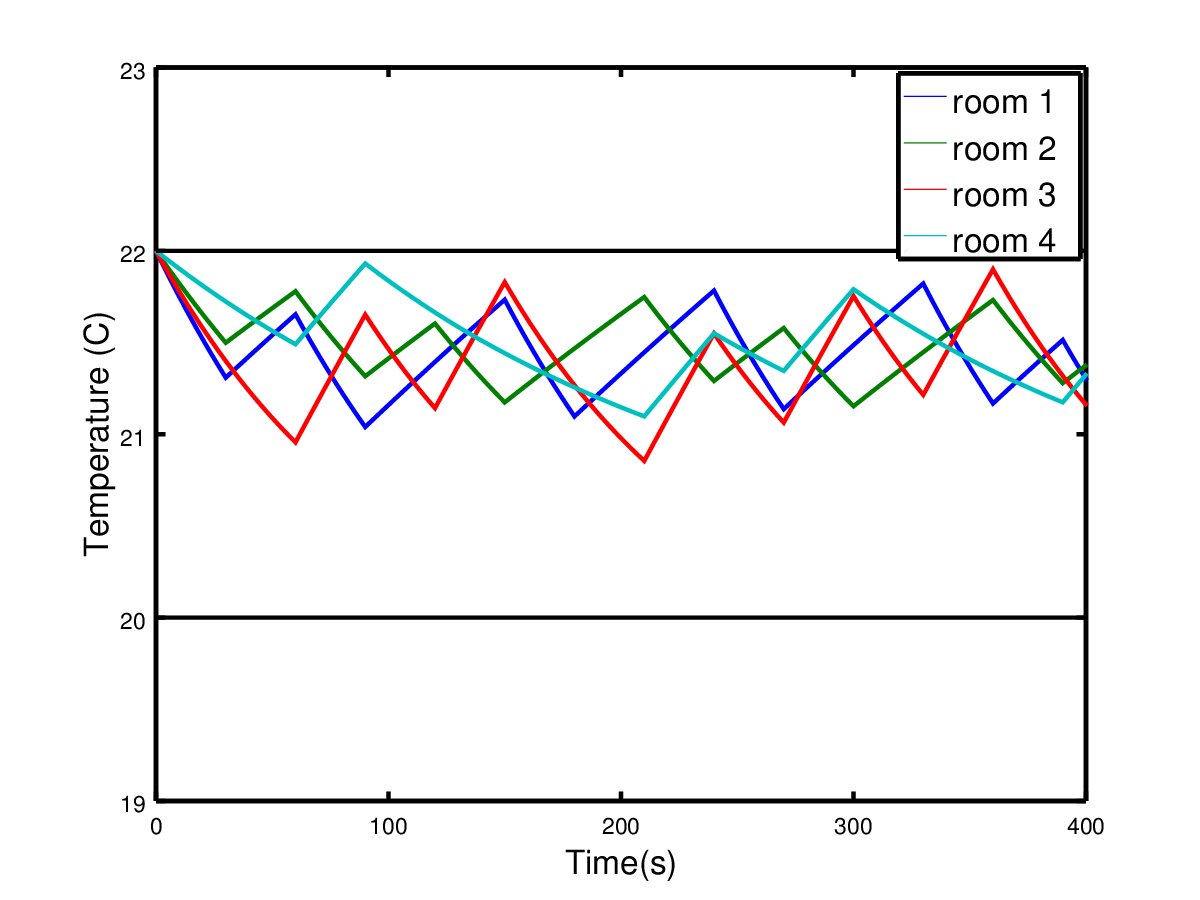}
&
 \includegraphics[width=0.4\linewidth,clip,trim=1cm 0cm 1cm 0cm]{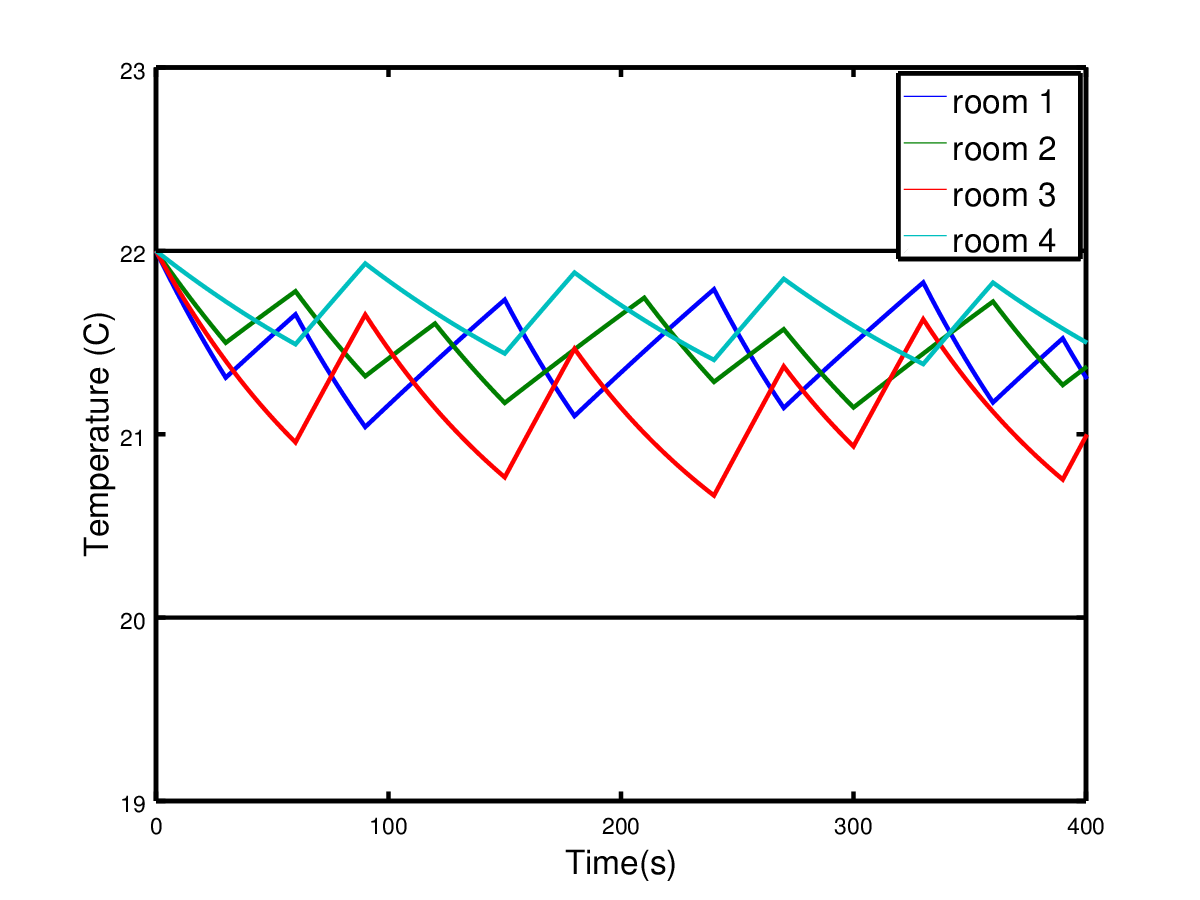}
\end{tabular}
\caption{Simulation of the four-room case study with our synthesis method (left) and with the synthesis method of \cite{NL_minimator}  (right).}
\label{fig:simu}
\end{figure}

\subsection{DC-DC converter}

This linear example is taken from \cite{beccuti2005optimal} and has
already been treated with the state-space bisection method in a linear
framework in \cite{fribourg2014finite}.

The system is a boost DC-DC converter with one switching cell.  There
are two switching modes depending on the position of the switching
cell. The dynamics is given by the equation $\dot x (t) =
A_{\sigma(t)} x(t) + B_{\sigma(t)}$ with $\sigma(t) \in U = \{ 1,2
\}$. The two modes are given by the matrices:

$$ A_1 = \left( \begin{matrix}
          - \frac{r_l}{x_l} & 0 \\ 0 & - \frac{1}{x_c} \frac{1}{r_0 + r_c}
         \end{matrix} \right)  \quad B_1 = \left( \begin{matrix}
         \frac{v_s}{x_l} \\ 0 \end{matrix} \right) $$

$$ A_2 = \left( \begin{matrix} - \frac{1}{x_l} (r_l +
  \frac{r_0.r_c}{r_0 + r_c}) & - \frac{1}{x_l} \frac{r_0}{r_0 + r_c}
  \\ \frac{1}{x_c}\frac{r_0}{r_0 + r_c} & - \frac{1}{x_c}
  \frac{r_0}{r_0 + r_c}
         \end{matrix} \right)  \quad B_2 = \left( \begin{matrix}
         \frac{v_s}{x_l} \\ 0 \end{matrix} \right)  $$

with $x_c = 70$, $x_l = 3$, $r_c = 0.005$, $r_l = 0.05$, $r_0 = 1$,
$v_s = 1$.  The sampling period is $\tau = 0.5$.  The parameters are
exact and there is no perturbation.  We want the state to return
infinitely often to the region~$R$, set here to $\lbrack 1.55 , 2.15
\rbrack \times \lbrack 1.0 , 1.4 \rbrack$, while never going out of
the safety set $S = \lbrack 1.54 , 2.16 \rbrack \times \lbrack 0.99 ,
1.41 \rbrack$.
On this example, the Euler-based method {\em fails} while {\em DynIBEX} succeeds
rapidly.

\begin{table}[h]
 \centering
\begin{tabular}{|c|c|c|}
\hline 
 &\multicolumn{1}{c|}{Euler} & \multicolumn{1}{c|}{DynIBEX} \\
\hline
$R$ & \multicolumn{2}{c|}{$[1.55,2.15]\times[1.0,1.4]$} \\
$S$ & \multicolumn{2}{c|}{$[1.54,2.16]\times[0.99,1.41]$} \\
\hline
$\tau$ &\multicolumn{2}{c|}{ 0.5 }\\
\hline
 Complete control & No & Yes\\
\hline
$\lambda_1$  & $-0.014215$ &\\
$\lambda_2$  & $0.142474$ &\\
$C_{1}$  & $6.7126 \times 10^{-5}$ &\\
$C_{2}$ &  $2.6229 \times 10^{-2}$ &\\
\hline
Number of balls/tiles & x & 48 \\
\hline
Pattern length & x & 6 \\
\hline
CPU time & x & < 1 second \\ \hline
 \end{tabular}
\label{table:DC}
\caption{Numerical results for the DC-DC converter example.}
 \end{table}

 \subsection{Polynomial example}
 We consider the polynomial system taken from \cite{liu2013synthesis}:
\begin{equation}
 \left \lbrack \begin{matrix}
  \dot x_1 \\ \dot x_2
 \end{matrix} \right \rbrack  =
 \left \lbrack \begin{matrix} -x_2 - 1.5 x_1 - 0.5 x_1^3 + u_1 \\ x_1 + u_2 
   \end{matrix} \right \rbrack.
\end{equation}
The control inputs are given by $u = (u_1,u_2) =
K_{\sigma(t)}(x_1,x_2)$, $\sigma(t) \in U = \{ 1,2,3,4 \}$, which correspond to
four different state feedback controllers $K_1(x) = (0,-x_2^2 + 2)$,
$K_2(x) = (0,-x_2)$, $K_3(x) = (2,10)$, $K_4(x) = (-1.5,10)$.  We thus
have four switching modes. The disturbances are not taken into account.
The objective is to visit infinitely often {\em two} zones $R_1$ and $R_2$,
without going out of a safety zone $S$.

 \begin{table}[h]
 \centering
\begin{tabular}{|c|c|c|}
 \hline 
 &\multicolumn{1}{c|}{Euler} & \multicolumn{1}{c|}{DynIBEX} \\
\hline
 $R_1$ & \multicolumn{2}{c|}{$[-1,0.65]\times[0.75,1.75]$} \\
 $R_2$ & \multicolumn{2}{c|}{$[-0.5,0.5]\times[-0.75,0.0]$ }\\
 $S$ &  \multicolumn{2}{c|}{$[-2.0,2.0]\times[-1.5,3.0]$ }\\
 \hline
$\tau$ & \multicolumn{2}{c|}{0.15} \\
\hline
Time subsampling & $\tau/20$ & \\
\hline
 Complete control & Yes & Yes \\
\hline
$\lambda_1$  & $-1.5$   &   \\
$\lambda_2$  &  $-1.0$ &\\
$\lambda_3$  &  $-1.1992 \times 10^{-8}$ & \\
$\lambda_4$ &  $-5.7336 \times 10^{-6}$ & \\
$C_{1}$  &  641.37     &    \\
$C_{2}$ &  138.49 &\\
$C_{3}$  &  204.50 & \\
$C_{4}$ & 198.64 &\\
\hline
Number of balls/tiles & 16 \& 16 & 1 \& 1 \\
Pattern length & 8  & 7  \\
\hline
CPU time & 29 \& 4203  seconds & <0.1 \& 329 seconds \\ \hline
  \end{tabular}
\label{table:PE}
\caption{Numerical results for the polynomial example example.}
 \end{table}

For Euler and {\em DynIBEX}, the table indicates {\em two} CPU times corresponding to the reachability from $R_1$
to $R_2$ and vice versa.
On this example, the Euler-based method is much slower than {\em DynIBEX}.
 
 \subsection{Two-tank system}
 
 The two-tank  system  is a linear example taken from \cite{hiskens2001stability}. The system consists of two tanks and two valves.
 The first valve adds to the inflow of tank 1 and the second valve is a drain valve for tank 2. 
 There is also a constant outflow from tank 2 caused by a pump. The system is linearized at a desired
 operating point. The objective is to keep the water level in both tanks 
 within limits using a discrete open/close switching strategy for the valves. 
 Let the water level of tanks 1 and 2 be given by $x_1$ and $x_2$ respectively. 
 The behavior of $x_1$ is given by $\dot x_1 = -x_1 - 2$ when the tank 1 valve is closed, 
 and $\dot x_1 = -x_1 + 3$ when it is open. Likewise,
 $x_2$ is driven by $\dot x_2 = x_1$ when the tank 2 valve is closed and $\dot x_2 = x_1 - x_2 - 5$ when it 
 is open. 
On this example, the Euler-based method works better than {\em DynIBEX}
in terms of CPU time.

 \begin{table}[h]
 \centering
\begin{tabular}{|c|c|c|}
 \hline 
 &\multicolumn{1}{c|}{Euler} & \multicolumn{1}{c|}{DynIBEX} \\
\hline
$R$ & \multicolumn{2}{c|}{$[-1.5,2.5]\times[-0.5,1.5]$} \\
$S$ & \multicolumn{2}{c|}{$[-3,3]\times[-3,3]$} \\
\hline
$\tau$ & \multicolumn{2}{c|}{0.2} \\
\hline
Time subsampling & $\tau/10$ & \\
\hline
Complete control & Yes & Yes \\
\hline
$\lambda_1$  & 0.20711     &       \\
$\lambda_2$  &  -0.50000 &\\
$\lambda_3$  &  0.20711 &  \\
$\lambda_4$ &  -0.50000 & \\
$C_{1}$  &  11.662    &             \\
$C_{2}$ & 28.917&\\
$C_{3}$  &  13.416 &\\
$C_{4}$ & 32.804& \\
\hline
Number of balls/tiles & 64 & 10 \\
Pattern length & 6 & 6 \\
\hline
CPU time & 58 seconds & 246 seconds\\ \hline
  \end{tabular}
\label{table:TT}
\caption{Numerical results for the two-tank example.}
 \end{table}

 \subsection{Helicopter}
 
 The helicopter is a linear example taken from \cite{ding2011reachability}. The problem is to control a quadrotor helicopter toward 
 a particular position on top of a stationary ground vehicle, while satisfying constraints 
 on the relative velocity. 
Let $g$ 
 be the gravitational constant, $x$ (reps. $y$) the position 
 according to $x$-axis (resp. $y$-axis), $\dot x$ (resp. $\dot y$) the velocity according to $x$-axis (resp. $y$-axis),
 $\phi$ the pitch command and $\psi$ the roll command. 
 The possible commands for the pitch and the roll are 
 the following: $\phi,\psi \in \{ -10,0,10 \}$.
 Since each mode corresponds to a pair $(\phi,\psi)$, there are nine switched modes.
 The dynamics of the system is given by the equation:
 $$ \dot X = \begin{pmatrix}
              0 & 1 & 0 & 0 \\ 
              0 & 0 & 0 & 0 \\ 
              0 & 0 & 0 & 1 \\               
              0 & 0 & 0 & 0               
              \end{pmatrix} X + \begin{pmatrix}
              0 \\ g\sin(-\phi) \\ 0 \\ g\sin(\psi) \end{pmatrix}              
 $$
where $X = ( x \ \dot x \ y \ \dot y)^\top$. Since the variables $x$ and $y$
are decoupled in the equations and follow the same equations (up to the sign of the command), it suffices
to study the control for $x$ (the control for $y$ is the opposite).
On this example again, the Euler-based method works better than {\em DynIBEX}
in terms of CPU time.

 \begin{table}[h]
 \centering
\begin{tabular}{|c|c|c|}
   \hline 
   &\multicolumn{1}{c|}{Euler} & \multicolumn{1}{c|}{DynIBEX} \\
   \hline
   $R$ & \multicolumn{2}{c|}{$[-0.3,0.3]\times[-0.5,0.5]$} \\
   $S$ & \multicolumn{2}{c|}{$[-0.4,0.4]\times[-0.7,0.7]$} \\   
\hline
$\tau$ & \multicolumn{2}{c|}{0.1} \\
\hline
Time subsampling & $\tau/10$ & \\   
   \hline
 Complete control & Yes  & Yes \\
\hline
$\lambda_1$  & 0.5    &        \\
$\lambda_2$  &  0.5 &\\
$\lambda_3$  &  0.5 &  \\
$C_{1}$  &  1.77535 &                     \\
$C_{2}$ & 0.5 &\\
$C_{3}$  &   1.77535& \\
\hline
Number of balls/tiles & 256 & 35 \\
Pattern length & 7 & 7 \\
\hline
CPU time &  539 seconds & 1412 seconds\\ \hline
  \end{tabular}
\label{table:HM}
\caption{Numerical results for the helicopter motion example.}
 \end{table}

\subsection{Analysis and comparison of results}

Our method presents the advantage over the work of \cite{NL_minimator} that
no numerical integration is required for the control synthesis. 
The computations just require the evaluation of given 
functions $f_j$ and
(global error) functions $\delta_j$ at sampling times. The synthesis 
is thus 
{\em a priori} cheap compared to the use of numerical integration schemes (and even compared to
exact integration for linear systems).
However, most of the computation time is actually taken by the search
for an appropriate radius $\delta$
of the balls $B_i$ ($1\leq i\leq m$) that cover $R$, and the search for appropriate 
patterns $\pi_i$ that make the trajectories issued from $B_i$ return to~$R$.

We observe on the examples that the resulting control strategies synthesized by our method are quite different from those obtained by the Runge-Kutta method of \cite{NL_minimator} (which uses in particular rectangular tiles instead of balls). 
This may explain why the experimental results are here contrasted:
Euler's method works better on 3 examples and worse on the 2 others.
Besides the Euler method fails on one example (DC-DC converter) while {\em DynIBEX} succeeds 
on all of them. Note however that our Euler-based implementation is made
of a few hundreds lines of interpreted code Octave while {\em DynIBEX} is made of around five thousands of compiled code C++.

\section{Final Remarks}\label{sec:fr}
We have given a new Euler-based method for controlling sampled switched systems,
and compared it with the Runge-Kutta method of \cite{NL_minimator}. The method is remarkably simple and gives already promising results. In future work, we plan to explore
the use of the {\em backward} Euler method instead of the forward Euler method used here
(cf: \cite{Beyn2010}).
We plan also to give general sufficient conditions ensuring the convexity
of the error function $\delta_j(\cdot)$; this would allow us to get rid of
the convexity tests that we perform so far numerically for each pattern.

\vspace{1em}
{\bf Acknowledgement.}
We are grateful to Antoine Girard, Jonathan Vacher, Julien Alexandre dit Sandretto and Alexandre Chapoutot for numerous helpful discussions.
This work has been partially supported by Federative Institute Farman (ENS Paris-Saclay and CNRS FR3311).

\bibliographystyle{eptcs}
\bibliography{minimator}

\end{document}